\documentclass[pra,superscriptaddress,a4paper,twocolumn,showpacs ]{revtex4}
\usepackage{amsmath,amsfonts,bbm, amsthm, graphicx,color}
\def\>{\rangle}
\def\<{\langle}

\def\E{ {\cal E} }

\def\H{ {\cal H} }
\def\G{ {\cal G} }
\def\M{ {\cal M} }
\def\N{ {\cal N} }

\def\D{ {\cal D} }

\def\S{ {\cal S} }
\def\I{ \mathbbm{1} }
\def\Tr{ \mbox{Tr} }
\def\non{ \nonumber\\}

\def\QRF{ | QRF (\alpha, \beta, \delta ) \> }

\newcommand{\J}[1]{\mathbf{J}_{#1} }

\newtheorem{theorem}{Theorem}[section]

\newtheorem{lemma}[theorem]{Lemma}

\begin{document}

\title{Optimal primitive reference frames.} %

\author{David Jennings}%
\affiliation{Institute for Mathematical Sciences, Imperial College London, London SW7 2BW, United Kingdom}%
\affiliation{QOLS, The Blackett Laboratory, Imperial College London, Prince Consort Road, SW7 2BW, United Kingdom}

\date{\today}

\begin{abstract}
We consider the smallest possible directional reference frames allowed and determine the best one can ever do in preserving quantum information in various scenarios. We find that for the preservation of a single spin state, two orthogonal spins are optimal primitive reference frames, and in a product state do approximately $22 \%$ as well as an infinite-sized classical frame. By adding a small amount of entanglement to the reference frame this can be raised to  $2(2/3)^5=26 \%$. Under the different criterion of entanglement-preservation a very similar optimal reference frame is found, however this time for spins aligned at an optimal angle of $87$ degrees. In this case $24\%$ of the negativity is preserved. The classical limit is considered numerically, and indicates under the criterion of entanglement preservation, that $90$ degrees is selected out non-monotonically, with a peak optimal angle of $96.5$ degrees for $L=3$ spins. 
\end{abstract}

\maketitle

\section{Introduction}

In information processing it is generally assumed that a large, classical background reference frame is defined and readily available. For example, to measure the spin of a particle along a certain axis, a Cartesian frame is required. 

In flat space, with Euclidean geometry, it is usual and convenient to work with orthogonal axes. However, for a classical reference frame (CRF) any non-degenerate coordinate system, together with the usual vector addition law, is assumed to be operationally sharp in specifying any direction in space. Furthermore, the principle of general covariance \cite{wald} dictates that the laws of physics do not contain privileged sets of vector fields. As a result of these two assumptions it follows that any two non-degenerate CRFs will do equally well in specifying directions in space, and so \emph{there is no operational distinction between non-degenerate classical reference frames.}

However, ``information is physical'', and so any directional CRF is unavoidably described by a classical system in a well-defined state. Beyond classical scenarios, situations arise where such an idealized CRF is not available and a finite-sized reference frame is required. For such a quantum reference frame (QRF), encoded say in the quantum state of a set of spin particles, one obtains deviations from the classical situation and finite-size effects can be important \cite{BRSreview, BRST2009, Mehdi2010, PY, CHRB1, cryptopowerprivateQRF}.

A typical scenario of interest might consist of an agent Eve who prepares a singlet spin state and gives the two halves to parties Alice and Bob, who wish to detect entanglement present in the system via their local measurements. In order to do so, they must share some notion of `up and down'. This classically corresponds to a shared Cartesian frame, and is described by a rotation $R(\Omega)$ that takes some local orthonormal triad for Alice or Bob to some local orthonormal triad for Eve, whether Alice and Bob decide to use orthonormal coordinates or not.
In the absence of a shared classical frame, the directionality must be encoded in quantum systems $QRF_A$ and $QRF_B$ that accompany the two halves of the singlet state. The quantum information, here entanglement, is only partially preserved in the relational properties of the composite state. The degree to which it is preserved will depend on both the size of the QRFs and on their particular directional states.

It is a remarkable fact that Nature provides a minimum non-zero size of spin for particles, e.g.\ the intrinsic spin of an electron, and hence a minimum scale for a directional reference frame. In three dimensions equipped with a Euclidean inner product we have the notion of a cross-product \cite{frankel}, and so classically it suffices to specify two linearly independent directions in space, $\mathbf{e_1} $ and $\mathbf{e_2}$, with a third direction then taken to be $\mathbf{e_1 \times e_2}$. In the quantum regime, we then find that the simplest possible QRF allowed consists of two non-aligned spin-1/2 particles \cite{BRSreview}.

In this paper we consider the fundamental limits of directional reference frames, and consider a QRF that consists of just two spin particles. Our primary interest is in the properties of the quantum reference frame itself, separate from the system for which it is used, and so we take the QRF to be in a pure state, uncorrelated with any other system. While classically any non-degenerate coordinate system performs as well as another, this is not the case for a quantum reference frame. Intuitively one would expect that a QRF with spins aligned orthogonally is best, since it is ``farthest from the spins being parallel'', surprisingly however, we shall find that this is not generally the case. Instead we find that entanglement can assist the QRF, and that orthogonal spins are \emph{almost} special.

\subsection{Overview of main results}

In section \ref{primitive} we introduce the primitive quantum reference frames and use the symmetries of G-twirling to restrict to a canonical set of states.  

In section \ref{singlespin} we determine the optimal primitive reference frame that preserves a single spin state and find that the best one can ever do is a $26 \%$ reduction of the volume of state space, for orthogonal reference frame spins with entropy of entanglement $E=0.19$. 

Section \ref{entanglement} then deals with the operational criterion of entanglement preservation, in which two distinct cases arise: either Alice lacks a directional frame or both Alice and Bob lack a directional reference frame. We establish a useful theorem for the negativity $\N$ of a bipartite state and analyse the loss in entanglement in terms of a sub-normalized quantum operation on one half of the entangled state. Surprisingly, we find that the optimal alignment of the product reference frame spins is at an angle of $87$ degrees, which preserves $23.6 \%$ of the negativity, and by adding some entanglement to the reference frame this can be raised to $24.4 \%$. We also find that it is possible to obtain greater negativity if we use a less than maximally entangled state. For the case of both Alice and Bob lacking a reference frame we find that the most entanglement one can preserve with primitive reference frames is $\N =7\%$, with identical frames for Alice and Bob, but now with reference frame spins aligned at an angle of $83$ degrees.

In section \ref{classicallimit} we consider how increasing the spin from the primitive reference frame scale to the classical limit affects matters. We show how classical vector addition emerges sharply in the large-$L$ limit, analyse how well the spin-L reference does at preserving entanglement and give evidence that suggests that $90$ degrees is asymptotically optimal \footnote{Although of course in this limit any linearly independent frame of spins preserves entanglement perfectly.}.
 
We conclude with section \ref{conclusions} and provide some technical details on G-twirling in the appendix.

\section{The Primitive Quantum Reference Frame}\label{primitive}
We look at the simplest possible quantum reference frame for specifying spatial directions in 3 dimensions - two non-aligned spin-1/2 particles. We refer to any QRF consisting of two spin-1/2 particles simply as a \textit{primitive quantum reference frame}. Our task is to determine the optimal primitive quantum reference frame for both state preservation and entanglement preservation. 
The absence of a CRF is described mathematically by the action of G-twirling, which involves an averaging of a quantum state over the rotation group. The technicalities of G-twirling as they relate to this work are described in the appendix.

\subsection{Canonical reference frame states}
In what follows we shall assume that the QRF is in a pure state $|QRF \>$. The system spin states are labelled with respect to a defining, background $CRF$ so that $|J_z=+1/2\> \equiv |0\>$ and $|J_z=-1/2\> \equiv |1\>$ and we order the three spins so that $\H = \H_{QRF} \otimes \H_{\rm system}$. 

For any single spin unitary $U$, and any composite state $\rho$ on $\H$ we can use $\G[\rho]=\G[ U^{\otimes 3} \rho (U^\dagger)^{ \otimes 3} ]$ to deduce that
\begin{eqnarray}
\G [ U^{\otimes 2} \otimes \I \rho (U^\dagger)^{ \otimes 2} \otimes \I] &=&\G[\I^{\otimes 2} \otimes U^\dagger \rho \I^{\otimes 2} \otimes U]
\end{eqnarray} 
and so a rigid rotation of the QRF is equivalent to an application of the inverse rotation to the system of interest, and does not affect the issue of irreversible loss of quantum information. This equivalence under the two different transformations is simply the equivalence between active and passive transformations in physics \cite{wald}. In particular this means that, without loss of generality, we may restrict ourselves to a subset of canonical QRF states. 

We can thus work in an external frame in which the Schmidt decomposition of the canonical QRF state is parameterized as 
\begin{eqnarray}\label{canonical}
|QRF (\alpha, \beta,\delta ) \> &=& \cos\alpha |0\> \otimes( \cos \frac{\beta}{2}|0\> + e^{i\delta } \sin \frac{\beta}{2} |1\>) \non 
 &&\hspace{-1cm} + \sin \alpha |1\> \otimes ( \sin \frac{\beta}{2} |0\> - e^{-i\delta } \cos \frac{\beta}{2} |1\>).
\end{eqnarray}

\section{Optimal for a single spin state}\label{singlespin}

The G-twirling of 3 spins under rotations, results in a single protected subsystem qubit, and in what follows we shall consider the encoding of a single spin state into this protected subsystem. The technical details of G-twirling are contained in appendices \ref{Alack} and \ref{Aprotected}, while appendix \ref{virtobs} describes the spin observables that allow access to the virtual subsystem degrees of freedom. 

The key property of the G-twirling map on the 3 spin system is that it splits up into a fully decohering map on the 4-dimensional subspace corresponding to a total angular momentum $J=3/2$, and a \emph{partially} decohering map on the orthogonal $J=1/2$ subspace. The $J=1/2$ subspace can in turn be split into two virtual qubit systems $\M_2 \otimes \N_2$. More explicitly, we have
\begin{eqnarray}
\G[ \rho] &=& \D_1 [\Pi_1 \rho \Pi_1 ] + \D_2[ \Pi_2 \rho \Pi_2]
\end{eqnarray}
where $\Pi_1$ is the projector onto the $J=3/2$ subspace, and $\Pi_2$ is the projector onto the $J=1/2$ subspace. The operation $\D_1$ is fully decohering on the support of $\Pi_1$, while $\D_2 = \D_{\M_2} \otimes \I_{\N_2}$, which signifies that it decoheres fully on a 2-dimensional virtual subsystem $\M_2$ while leaving the 2-dimensional virtual subsystem $\N_2$ unaffected (see appendix \ref{Aprotected} for more details). 

We adopt the convention of $(\alpha, \beta, \delta)$ for the $QRF$ parameters as given in (\ref{canonical}), and use single spin state parameters $(\theta, \phi)$, in the state $|\theta, \phi \>=\cos \theta |0\> +e^{i \phi} \sin \theta |1\>$.

The encoding of the spin state into the protected subsystem defines a quantum channel from the physical system into a virtual subsystem. The quality of the QRF at protecting quantum information is then equivalent to how noisy this resultant quantum channel is, and as such we do not worry about unitary rotations in the encoding of the state. We also note that the reduction in purity under the encoding can vary greatly with the direction along which the spin is polarized. To eliminate scenarios in which the state is strongly dephased (e.g.\ ones that preserve only classical information) we should require that the encoding of the quantum state does not decohere strongly along any one direction.

In light of this, we use  the volume of the image of the state space under the mapping induced by the G-twirling as a suitable measure of how well the QRF preserves a general spin state. A large volume for the image of the Bloch sphere implies a high average purity for all spin polarizations, and conversely an encoding that preserves a high level of purity for all polarizations will result in a large image volume. Furthermore, since the trace distance coincides with the Euclidean distance between Bloch vectors \cite{NC}, we will simply use the Euclidean volume element as our measure.
Given an image volume $\tilde{V}$ we can define a characteristic distance $r = (\frac{3 \tilde{V}}{4\pi})^{1/3}\le1$ for the encoding, which gives a measure for the average distinguishability of two orthogonal states under the encoding.

We first consider a canonical QRF that is in a product state 
\begin{eqnarray}\label{prodqrf}
|QRF\>= |0\> \otimes( \cos \frac{\beta}{2}|0\> + e^{i\delta } \sin \frac{\beta}{2} |1\>)
\end{eqnarray}
with the full state $|\Psi \>=|QRF\> \otimes |\theta, \phi \>$ then subject to G-twirling. Only the reduced state on the protected virtual subsystem is unaffected, and we may obtain the Bloch vector in the protected virtual subsystem by simply computing the qubit state $\Tr_{\M_2} [ \Pi_2 |\Psi \> \< \Psi | \Pi_2 ]$. 

The bases for the $J=3/2$ and $J=1/2$ are $\{ |J, s, p\> \}$, and are given explicitly in appendix \ref{Aprotected}. We define a virtual set of two qubits for the $J=1/2$ sector as $|\bar{i}\> \otimes| \bar{j} \> := |\frac{1}{2}, i, j \>$. In this notation, the second of the two virtual qubits is the protected virtual subsystem in which the quantum state is stored.

To determine the protected qubit state we first project into the $J=1/2$ sector and then trace out the first virtual qubit. Upon doing this for the uncorrelated QRF (\ref{prodqrf}), we find that the Bloch vector of the protected state is given by $\mathbf{R} =(R_x, R_y,R_z)$ where 
\begin{eqnarray}
 R_x&=&\frac{1}{2 \sqrt{3}}(-2 \cos 2 \theta  \sin^2 \frac{\beta }{2}+\cos(\delta -\phi ) \sin \beta  \sin 2 \theta ) \non  
R_y&=& -\frac{1}{2 \sqrt{3}} \sin \beta  \sin 2 \theta  \sin (\delta -\phi ) \non
R_z &=& \frac{1}{6} (\cos \beta  (-2+\cos 2 \theta )
 \non && +\cos 2 \theta +\cos (\delta -\phi ) \sin \beta  \sin 2 \theta ). 
\end{eqnarray}

For two aligned spins we have $\beta=0$ and the Bloch vector images get set to $(0,0, -\frac{2}{3} \sin^2 \theta)$ and so is a very poor image of the Bloch sphere. Similarly for anti-aligned spins, $\beta=\pi$, we get  $(-\frac{1}{\sqrt{3}} \cos 2 \theta ,0,\frac{1}{3})$ and so is also a very poor image of the Bloch sphere; both are 1-dimensional images of the sphere. 

As already mentioned, we can view the action of the G-twirling as defining a single qubit quantum channel. The mapping induced by the G-twirling can then be written as an affine map of the state space $\mathbf{x} \rightarrow A(\alpha, \beta, \delta) \mathbf{x} +\mathbf{b}(\alpha,\beta ,\delta)$, composed of a linear transformation followed by a global translation along a fixed direction. The volume distortion factor is obtained from the determinant of the linear transformation and for a canonical product state QRF is found to equal $\rm{det} A = \frac{2}{9} \sin^2 \beta$, which has a peak for $\beta_{\rm opt} =90$ degrees. Thus, the best product QRF for the preservation of the single spin logical state is obtained when we orient the two spins orthogonally, for which the Bloch sphere is shrunk by about $22\%$. The image of the Bloch sphere becomes
\begin{eqnarray}
\mathbf{R}(\theta,\phi)&=&\left(
\begin{array}{c}
 \frac{1}{2 \sqrt{3}}(- \cos 2 \theta +\cos(\delta-\phi) \sin 2 \theta)  \\
 \frac{1}{\sqrt{3}}\sin 2\theta \sin(\delta-\phi)  \\
 \frac{1}{6} (\cos 2 \theta +\cos (
\delta-\phi ) \sin 2 \theta )
\end{array}
\right).\nonumber
\end{eqnarray}
Note that the arbitrary phase angle on the QRF acts simply as a linear translation of the phase of the encoded state.

For the most general QRF in which we allow the two spins to be entangled, it is possible to once again calculate the volume distortion of the state space. In this case
\begin{eqnarray}
\rm{det} A (\alpha, \beta) &=& \frac{2}{9} (\cos \alpha - \sin \alpha)^4 (\cos \alpha +\sin \alpha )^2 \sin^2 \beta \nonumber
\end{eqnarray}
where for simplicity we have set $\delta=0$ as this turns out to be optimal. We now have a maximum volume for $\beta_{\rm opt} = 90$ degrees and $\alpha_{\rm opt} = \arctan (2\sqrt{2}-3 )$. This corresponds to an entropy of entanglement \cite{ententanglement} of $E[|QRF\>\<QRF|]=0.19$ for the QRF.

Thus, the best possible primitive quantum reference frame for the encoding a single spin state shrinks the Bloch Ball by a factor $\frac{64}{243} \approx 26 \%$, or roughly to a radius of $0.64$, and shows that an orthogonal frame is best, but a little bit of entanglement also helps.

\section{Optimal for spin entanglement}\label{entanglement}
We now determine the best primitive QRF for the preservation of entanglement. As mentioned before we have a global freedom to rigidly rotate the QRF in space. We consider a composite state $|\Psi \> =|QRF_A \> \otimes |\varphi \> \otimes |QRF_B \>$, where $|\varphi\> $ is a maximally entangled state of two spin-1/2 particles, however rotation of the two local QRFs is equivalent under G-twirling to the application of rotation of the two spins that make up $|\varphi\>$. Since every such pure maximally entangled state is locally equivalent to the singlet state, we can take $|\varphi \> = |\psi ^-\>$ and work with a general QRF, or alternatively we can fix $|QRF_A\> $ and $|QRF_B \>$ to be canonical QRF states and leave $|\varphi \>$ unspecified, but maximally entangled. We opt for the latter.

As the measure of bipartite entanglement, we use negativity \cite{negativity}, which is defined as $\N [ \rho_{AB} ] := \frac{1}{2} ( || \rho_{AB}^{T_B} ||_1 - 1)$ where $||M||_1 :=\Tr \sqrt{M^\dagger M}$ and $T_B$ denotes partial transpose with respect to $B$. While the negativity is usually convenient to calculate, the spectrum of $\rho_{AB}^{T_B}$ does not admit an analytic expression for the most general QRF, and so for such cases numerics will be needed.

There are two distinct situations of interest. The first is where either Alice or Bob does not share a classical reference frame with Eve; the second is where neither Alice nor Bob share a CRF with Eve. Since the latter is obtained by two independent G-twirlings for Alice and Bob (see appendix \ref{Atwirling}), we shall first consider the situation where Alice lacks the appropriate CRF. Moreover, this is closest to our main aim of determining the best primitive QRF, since this situation tests the quality of a single QRF, as opposed to the joint functioning of two such reference frames.

\subsection{Twirling Alice}

We consider the state consisting of a canonical QRF for Alice and a maximally entangled two spin state. The four spins are in the state
\begin{eqnarray}
| \Psi \> &=& \QRF \otimes |\varphi\>.
\end{eqnarray}
G-twirling Alice involves $\Psi \equiv |\Psi \> \< \Psi |$ transforming as 
\begin{eqnarray}
\Psi \rightarrow \G[ \Psi ]=\D_1 [\Pi_1 \Psi \Pi_1] + \D_2 [\Pi_2 \Psi \Pi_2] 
\end{eqnarray}
 which acts only on the first three spins. 

Since measures of entanglement are unitarily invariant, we define a unitary $U$ from the physical basis for the three spins to a virtual one, consisting of orthonormal basis vectors for the two sectors $\H_1$ and $\H_2$. We then analyse the two (unnormalized) states $p_1 \rho_1=U \D_1[\Pi_1 \Psi \Pi_1 ] U^\dagger $ and $p_2 \rho_2 = U \D_2[\Pi_2 \Psi \Pi_2] U^\dagger$, where we project into the two sectors of the first 3 spins and transform to the virtual basis. In the virtual basis we use the states $|\bar{0}\>$ and $|\bar{1}\>$ of the first spin to label the two orthogonal sectors.

The action of the G-twirling on the first three spins with probability $p_1$, results in the first `sector spin' being projected into the $|\bar{0}\>$ state and spins 2 and 3 in being fully depolarized, and with probability $p_2$, the sector spin being projected into $|\bar{1}\>$ and spin 2 fully depolarized. This leaves us with
\begin{eqnarray}
\rho_1 &=& \frac{1}{4}|\bar{0}\>\<\bar{0}|\otimes   \I_2 \otimes \I_3\otimes\sigma^{(1)}_{4} \non 
\rho_2 &=& \frac{1}{2}|\bar{1}\>\<\bar{1}|\otimes \I_2 \otimes\sigma^{(2)}_{34}  
\end{eqnarray}
(subscripts on the right-hand side label the qubits) and $\rho =\G[ \Psi ]=\rho = p_1 \rho_1 + p_2 \rho_2$ .
We find that for the singlet state the projection probabilities are given by
\begin{eqnarray}
p_{1,2}&=&\frac{1}{2} \pm \frac{1}{6}(\cos \beta \non &&-2 \cos \alpha  \cos \beta  \cos \delta  \sin \alpha  +\cos \delta  \sin 2 \alpha ).
\end{eqnarray}

We wish to maximize the negativity of the full state $\G[\Psi]$ across A and B. Since negativity is unitarily invariant we attack the problem in the virtual basis, and establish the following lemma and theorem, which are useful for our analysis.

\begin{lemma} Given hermitian operators $\{A_i \}$ and $\{B_i \}$, with each pair in $\{A_i \}$ having mutually orthogonal support, we have that $|\sum_i A_i \otimes B_i | = \sum_i |A_i| \otimes |B_i|$, where $|M|:=\sqrt{M^\dagger M}$ for any operator $M$.
\end{lemma}
\begin{proof}
By spectral decomposition, any hermitian operator $M$ can be written as $M= M_+ - M_-$, where $M_\pm$ are positive operators with orthogonal support, and thus $|M| = M_+ + M_-$. In particular, 
\begin{eqnarray}
A_i \otimes B_i &=& (A_{i,+}\otimes B_{i,+}+ A_{i,-}\otimes B_{i,-})\non &&-(A_{i,+}\otimes B_{i,-}+ A_{i,-}\otimes B_{i,+}).
\end{eqnarray}
However, $A_{i,\pm}\otimes B_{i,\pm}$ is orthogonal to $A_{k,\mp}\otimes B_{k,\pm}$ for $i=k$ by definition of the decomposition into positive operators, and orthogonal for $i\ne k$ by assumption on the set $\{A_i \}$. Hence, we have $X:=\sum_i A_i \otimes B_i = X_+ -X_-$ where $X_+ =\sum_{i, \pm} ( A_{i,\pm}\otimes B_{i,\pm})$ and $X_- =\sum_{i, \pm} ( A_{i,\mp}\otimes B_{i,\pm}),$ and so $|X| =|\sum_i A_i \otimes B_i| = \sum_i|A_i|\otimes |B_i|$. 
\end{proof}

\begin{theorem} For any bipartite state $\rho_{AB} = \sum_i p_i \varphi_{A,i} \otimes \sigma_{AB,i}$ with orthogonal states $\{\varphi_{A,i} \}$, the negativity of the state is given by $\N [\rho_{AB} ] = \sum_i p_i \N [\sigma_{AB,i} ]$.
\end{theorem}
\begin{proof} The negativity is defined as $\N (\rho_{AB}) := \frac{1}{2} ( || \rho_{AB}^{T_B} ||_1 - 1 )$, where $||X ||_1 := \Tr |X|$ and $T_B$ denotes partial transpose on $B$. Since $\rho_{AB}^{T_B} = \sum_i p_i \varphi_{A,i} \otimes \sigma^{T_B}_{AB,i}$, it follows from the previous lemma and basic trace properties that $\N(\rho_{AB})= \frac{1}{2}\sum_i p_i \Tr (\varphi_{A,i}) \Tr |\sigma^{T_B}_{AB,i}|-\frac{1}{2}= \sum_i p_i \N(\sigma_{AB,i})$.
\end{proof}

Due to the orthogonality of the different local sectors for $A$, the above theorem tells us that our task of maximizing the negativity of the full state is reduced to maximizing $ \N[ \rho_{AB}] = p_2 \N [\sigma^{(2)}_{34} ]$ over the virtual two spin output states $\sigma^{(2)}_{34}$. 

\subsection{Induced quantum operation on a single spin}
From the analysis in the previous section we established that it is sufficient to consider the induced \textit{selective} quantum operation on the one half of the maximally entangled state $|\varphi \>$, that takes $\varphi \rightarrow p_2 \sigma^{(2)}_{34}$. This operation depends on the specific QRF state $\QRF$ involved in the G-twirling and can be described by the linear completely positive map
\begin{eqnarray}
\E_A [ |\varphi \>\< \varphi | ] &=& M_1 |\varphi \>\< \varphi | M_1^\dagger +  M_2 |\varphi \>\< \varphi | M_2^\dagger \nonumber
\end{eqnarray}
where $M_{1,2}$ act non-trivially only on the first spin. The operators $M_{1,2}$ can be found by projecting into the $J=1/2$ sector and then tracing out the virtual subsystem $\M$. For the canonical primitive QRF state (\ref{canonical}) they are given by
\begin{eqnarray}\label{kraus}
M_1 &=& \left ( \begin{array}{c c} 
\frac{1}{\sqrt{2}}(e^{i \delta}\cos \alpha - \sin \alpha)  \sin \frac{\beta}{2} & 0 \\
-\frac{1}{\sqrt{6}}(e^{i \delta}\cos \alpha + \sin \alpha) \sin \frac{\beta}{2} & \sqrt{\frac{2}{3}} \cos \alpha \cos \frac{\beta}{2} \\
\end{array} \right ) \non 
M_2 &=& \left ( \begin{array}{c c} 
0 &\frac{1}{\sqrt{2}}(e^{i \delta}\cos \alpha - \sin \alpha)  \sin \frac{\beta}{2}  \\
0& \frac{1}{\sqrt{6}} (e^{i \delta}\cos \alpha + \sin \alpha) \sin \frac{\beta}{2} \\
\end{array} \right ).
\end{eqnarray}

We have fixed the QRF state to be canonical, and can now consider a maximally entangled spin state shared between $A$ and $B$ which, in general, can be written as $|\varphi\> = \I \otimes U |\psi ^-\>$, for some unitary $U$. The action of the quantum operation on this state is then $\E_A [ |\varphi \> \< \varphi|] =\sum_i (M_i\otimes \I) |\varphi\>\<\varphi | (M_i^\dagger \otimes \I) = (\I \otimes U) \E_A [ |\psi^- \>\< \psi^-|](\I \otimes U^\dagger)$ and from unitary invariance of negativity we conclude that all maximally entangled spin states $|\varphi \>$ will suffer exactly the same loss in negativity.

It is then a simple matter to compute $p_2\sigma^{(2)}_{34} = \E_A [ |\psi^-\>\< \psi^-|]$ and analyse its negativity. In particular, for the case of a canonical product state QRF, the state $\sigma^2_{34}$ is rank 2, and (for $ 0 \le \beta \le \pi$) has eigenvalues  $\lambda_{1,2} = \frac{1}{2} \pm \frac{\cos \frac{\beta}{2}}{\cos \beta - 3} $ with corresponding eigenstates 
\begin{eqnarray}
|e_1\> &=&  (-\sqrt{3} \cot \frac{\beta}{4} , -\sqrt{3}, -\tan \frac{\beta}{4}, 1) \non 
|e_2\> &=& (\sqrt{3}\tan \frac{\beta}{4} , -\sqrt{3}, \cot \frac{\beta}{2}, 1)
\end{eqnarray} in the virtual basis. 

For $\beta=\pi/2$ the negativity of the bipartite state may be found analytically and is $\N[\rho] = \frac{1}{3\sqrt{2}} \approx 23.57 \%$, and furthermore, to good approximation we have that
\begin{eqnarray}
\N(\beta,\gamma) \approx \frac{1}{3\sqrt{2}} |\sin \beta  \sin 2 \gamma |
\end{eqnarray}
for an entangled state $|\varphi \> =\cos \gamma |01 \> - \sin \gamma |10\>$. 
 Surprisingly however, it turns out that a state of orthogonal spins is \emph{not} the optimal product state QRF. It can be shown that one can do slightly better, and obtain $\N[\rho]=23.60 \%$ for $\beta_{\rm opt}=87$ degrees.

As in the previous section, the addition of a small bit of entanglement increases the performance of the QRF. Numerics on the full set of quantum reference frames show that the optimal primitive reference frame for preserving entanglement in a maximally entangled spin state achieves $\N[\rho]=24.4 \%$ for the QRF in the state with $\delta=0$, an entropy of entanglement of $E=0.167$ (for $\alpha_{\rm opt}=-0.15$) and the spins aligned at an angle $82$ degrees! 
Interestingly, this QRF is quite similar to the analytic optimal one we obtained for the quite different task of preserving a logical state. Furthermore, both do approximately $25\%$ as well as the classical limit, with the addition of entanglement providing a similar benefit in each case. 

It is also of interest to consider \textit{partially} entangled states $|\varphi\>$. If we use the entangled state $|\varphi \> = \cos \gamma |01\> -\sin \gamma |10\>$, instead of a maximally entangled state we find that the optimal alignment angle obeys
\begin{eqnarray}
\lim_{\gamma \rightarrow 0} \beta_{\rm opt} &=& \frac{\pi}{2}
\end{eqnarray}
however, more unusual is that we can preserve more output entanglement by using a \emph{less than maximally} entangled input state. A careful numerical analysis finds that the most entanglement that can be preserved with any QRF and any entangled two spin state is with $\beta_{\rm opt}=82$ degrees, $\delta_{\rm opt}=0$, $\alpha_{\rm opt}= -0.2$ and $\gamma_{\rm opt}=38$ degrees. For this we get $\N=25.2 \%$. Here we note that the negativity varies slowly for angles between about $80$ and $90$, while the addition of entanglement to the QRF contributes to a much larger variation for $\N$. The unusual parameters for the optimal QRF in each case arise from the finite-sized effects within the full Hilbert space, which occur once we impose the constraint that the QRF must be a separate, disentangled, system. For such small systems the states of the protected virtual subsystem only have partial overlap with those accessible product states.

\subsection{Twirling Alice and Twirling Bob}

The tools acquired in the previous section can be directly applied to the situation where neither Alice nor Bob has access to the reference frame of the prepared maximally entangled spin state. As before, we restrict to primitive quantum reference frames $QRF_A$ and $QRF_B$ for Alice and Bob respectively.

The composite system is composed of six spins, in the pure state $|\Psi \> =|QRF_A\> \otimes |\varphi \> \otimes |QRF_B \>$. However, the independent G-twirling at A and B transforms the state as $ \Psi \rightarrow \G_B[\G_A[\Psi]] \equiv \G [\Psi]$ where
\begin{eqnarray}
\G[ \Psi] &=& \sum_{i,j}\D_{B,i}[\D_{A,j}[\Pi_{i,j} \Psi \Pi_{i,j}]] 
\end{eqnarray}
where we again label local sectors as either $1$ or $2$, and define $\Pi_{i,j} \equiv \Pi_{B,i}\Pi_{A,j}$. 

We transform to the two local virtual bases, and use the orthogonality of the spin states labelling the sectors to deduce that $\N[\G[ \Psi]] = p_{2,2} \sigma^{(2,2)}_{3,4} $. For identical product reference frames we have $p_{2,2}= \frac{1}{18} (7-\cos \beta) \sin^2 \beta/2$ and $\sigma^{(2,2)}_{3,4} $ is rank 2. We obtain this state more directly via an application of the selective operations derived above, in other words
\begin{eqnarray}
p_{2,2} \sigma^{(2,2)}_{3,4} &=& \E_B [\E_A [|\varphi \> \< \varphi | ]]
\end{eqnarray}
with corresponding single spin Kraus matrices $\{M^A_i, M^B_j \}$ given by the local QRFs as in (\ref{kraus}). It is then straightfoward to determine, for example, that in the case of product state QRFs for both Alice and Bob, the largest negativity preserved under the two G-twirlings is only $\N =7\%$ for a QRF with spins aligned at $\beta_A =\beta_B= 83$ degrees.

\section{The Classical Limit: Spin-L quantum reference frames}\label{classicallimit}

We have obtained optimal angles close to 90 degrees for primitive QRFs, and so it is natural to ask what happens to these angles in the classical limit. To do this we consider a QRF composed of two spin-$L$ particles and analyse how $\beta_{\rm opt}$ varies with $L$ for a singlet state input. The sector structure at A now becomes $(\H_{2L}\oplus \H_{2L-1}\oplus \cdots \oplus \H_0 )\otimes \H_{1/2} = \H_{2L+1/2}\oplus 2\H_{2L-1/2}\oplus \cdots \oplus 2\H_{1/2}$, and we get $2L$ virtual protected qubits on $2L$ different sectors.

As before, we unitarily transform $A$ to a local `sector' plus `junk' system and a protected qubit system. The negativity of the resultant state (assuming singlet between A and B) is then given by
\begin{eqnarray}\label{negLeqn}
\N [ \G[\QRF \otimes|\psi^-\>]] &=& \sum_k p_k\N[ \sigma^k ]
\end{eqnarray}
where $k$ ranges over the $2L$ sectors containing protected qubits, $\sigma^k$ is the reduced state on the two spin subsystem, and $p_k$ is the projection probability of obtaining sector $k$.

We may once again obtain a reference frame dependent quantum operation $\E_A :\rho \mapsto  \sum_{i_k,k} M^k_{i_k} \rho M^{k \dagger}_{i_k}$ that describes the action of the G-twirling purely on Alice's spin and gives the required ensemble terms $p_k \sigma^k =\sum_{i_k} M^k_{i_k} |\psi^- \>\<\psi^-| M^{k \dagger}_{i_k}$ for the negativity.

\begin{figure}[t]
\centering
\includegraphics[width=7cm]{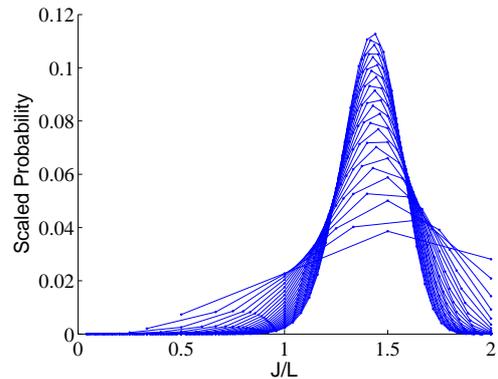}
\caption{[Color online] \textbf{Quantum Pythagoras' theorem:} getting $\sqrt{2}$ with spins up to $L=25$. The vertical axis is a rescaled probability to account for the increase in data points as we increase $L$. The horizontal axis is $J/L$ where $J$ is the measured total angular momentum of the two orthogonal spins.}
\label{Pythagoras}
\end{figure}

\subsection{Sectors for the Quantum Reference Frame}

To construct the unitary to go from the physical basis states to the $\H_k$ sector states we make use of the Clebsch-Gordan coefficients \cite{Rose, Edmonds} $\<L,L; m_1, m_2 | J,M\>$ to couple the two spin-$L$ particles in the QRF. 

In this setting it is illuminating to see how for large $L$ the state of two orthogonal spins $|J_z=L\> \otimes |J_x=L\>$ is distributed over the different sectors. Classically one expects the sharp vector addition law where the two orthogonal vectors add to one of length $\sqrt{2}L$. Numerics show that as $L$ increases the distribution of $|J_z=L\> \otimes |J_x=L\>$ is sharply peaked on the sector with total $J$ value closest to $\sqrt{2}L$ as one would expect (see figure \ref{Pythagoras}). For example, with $L=17$ it peaks on sector $J=24$ with $24/17 \approx 1.412$ being a good approximation to $\sqrt{2} \approx 1.414$. By increasing $L$ we probabilistically recover the standard vector addition, although for any finite-sized spin `$\sqrt{2}$' is always a rational number, which would have pleased the ancient Greeks. Classicality emerges through the distribution over the different sectors, and the value for the negativity of the twirled state is dominated by states $\sigma^{k_c}$ where $k_c$ labels a sector with total angular momentum near to the classical value.

For any angle of inclination $\beta$ for the QRF spins, we can construct the corresponding state 
\begin{eqnarray}
|QRF\> =|J_z=L\> \otimes|J_{\hat{n}(\beta)} =L\>
\end{eqnarray}
where $|J_{\hat{n}(\beta)} =L\>
$ is a coherent spin state polarized at an angle $\beta$ to the $Z$ axis, and then compute the negativity of $\G[ |QRF\> \otimes |\psi^-\>]$. Intuitively we expect the preserved negativity to increase with $L$ for optimal angles at each value of $L$. 

In figure \ref{optimal} we have computed the optimal angles $\beta_{\rm opt}(L)$, for a QRF composed of two spin $L$ particles, that preserve the most negativity in the singlet state. The numerical algorithm used is as follows: for a fixed $L$, we first construct Alice's unitary transformation $U_A(L)$ from the physical basis of the spin system $\H= \H_L \otimes \H_L \otimes \H_{1/2}$ to the virtual basis $\{ |s,j,p\> \}$ where $s$ is the sector label, $j$ labels the `junk' degrees of freedom affected by the G-twirling, while $p$ labels the degrees of freedom of the protected virtual subsystems. For any fixed angle $\beta$, we construct $|QRF (\beta)\>_A \otimes |\psi^-\>_{AB}$ and apply $U_A(L)\otimes \I_B$ to the full state. The negativity after G-twirling is then obtained by projecting onto sectors, tracing out onto the protected qubit systems and using (\ref{negLeqn}) to find the total preserved negativity between $A$ and $B$.

 We find that the function $\beta_{\rm opt} (L)$ rises to a peak of $96.5$ degrees for $L=3$ before starting a slow decline. For large values of $L$ we find that the function $\N(\beta, L= \rm{constant})$ flattens out over the interval $0$ to $180$ degrees and vanishes at the two end-points where the QRF becomes degenerate. The maximal preservation of negativity increases rapidly with $L$, reaching $90 \%$ already by $L=6.5$, see figure \ref{negativity}. 

Some intuition as to why small values of $L$ have $\beta_{\rm opt} > 90$ degrees can be gained from figure \ref{Pythagoras}. For small values of $L$ the probability distribution is asymmetric and has a relatively large weight on the sector with total angular momentum $2L$. This sector contains no protected subsystem, and so it is clearly beneficial to have an angle of inclination slightly above $90$ to reduce the contribution from a sector that can preserve no negativity. Of course there is another competing aspect. Having the spins strongly anti-aligned results in virtual qubit states $\sigma_k$ with poor negativity. Note that the $L=1/2$ case is in fact the only one with $\beta_{\rm opt} <90$ degrees, and so in this case it is not merely a matter of avoiding the sector of largest $L$, but there is also a strong dependence of the negativity for each $\sigma_k$ on the QRF spin alignment. In light of these results, it would be of interest in future work to study this emergence of classicality, and to further analyse the competing mechanisms at work for small systems.

\begin{figure}[t]
\centering
\includegraphics[width=7cm]{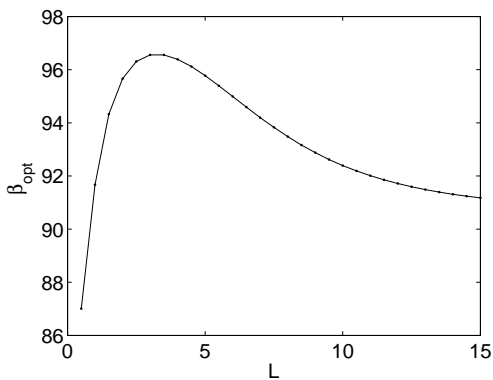}
\caption{\textbf{The classical limit:} The optimal angle of inclination $\beta_{\rm opt}$ for the two spins in the QRF as a function of $L$. A peak occurs of $\beta_{\rm opt}=96.5$ degrees at $L=3$. Numerical analysis for large $L$ implies that $\N (\beta)$ approaches a step function with zeros at $\beta=0$ and $\beta=180$ degrees, and that asymptotically $\beta_{\rm opt} \rightarrow 90$ degrees.}
\label{optimal}
\end{figure}
\begin{figure}[t]
\centering
\includegraphics[width=7cm]{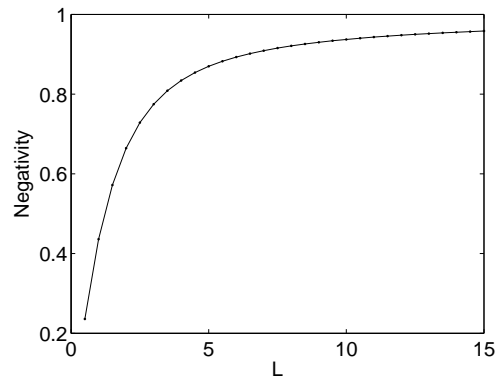}
\caption{\textbf{Negativity in the classical limit:} Preservation of negativity increases rapidly with increasing spin $L$. States $\sigma^{k_c}$ from sectors with angular momentum close to the classical value dominate the negativity.}
\label{negativity}
\end{figure}

\section{Conclusions}\label{conclusions}

In this work we have considered the fundamental limits of finite-sized quantum reference frames. We have analysed how well the most primitive such frame can ever perform under certain information preserving criteria, and have found a rough concordance as to the properties of the optimal primitive reference frame. The optimal frames involve spins \textit{roughly} orthogonal (between 82 and 90 degrees), and with a small degree of entanglement (an entropy of entanglement of about $0.15$). The actual performance of a single optimal reference frame in each case is roughly $25 \%$ of the infinite-limit classical frame. We also studied the classical limit of such a quantum reference frame and found that the usual vector addition law emerges gradually, and that a large amount of negativity can be preserved even for modest-sized spins. The optimal angle of inclination as a function of $L$ displays an unusual peak at $L=3$, before slowly decreasing towards $90$ degrees, where finite-size effects are washed out and the reference frame no longeer degrades the spin entanglement.

\begin{acknowledgments} DJ acknowledges the support of EPSRC, and thanks Terry Rudolph and Yeong-Cherng Liang for helpful discussions.
\end{acknowledgments}

\begin{appendix}

\appendix

\section{A lack of a classical reference frame}\label{Alack}

Here we describe how G-twirling arises when there is not a shared directional reference frame. We can imagine an agent Eve, preparing a system of spin particles that make up a quantum system in a state $|\Psi \>$. This pure quantum system is prepared in conjunction with a `classical reference frame' $CRF_E$ which sharply specifies directions in space. We could regard $CRF_E$ as a large spin-$j$ system in a highly coherent state, however here we simply take it as an abstract background setting and place it on the classical side of the Heisenberg cut.

Eve then sends the system to Alice, but unfortunately Alice does not share a private reference frame with Eve. Her local axes are related to Eve's by some unknown rotation and so Alice must average the state uniformly over all spatial rotations. The processed state represents the updated knowledge of the randomly rotated state, or the state from Alice's point of view, where she lacks knowledge of how her $CRF_A$ is related to Eve's $CRF_E$. In the case of multiple copies, instead of a single-shot procedure, Eve sends multiple copies of the same state to Alice, but if the systems are sent identically then the same unitary rotation is applied to each system then we have a perfect channel, where no averaging takes place. However if the systems are independently and randomly rotated around, then Alice must once again use an average over all rotations. These two perspectives correspond to the bayesian and frequentist views of quantum states.

In both cases the averaged state becomes
\begin{eqnarray}
\G[ | \Psi \> \< \Psi | ] \equiv \int d \Omega U(\Omega) | \Psi \> \< \Psi | U^\dagger (\Omega)
\end{eqnarray}
where we use the Haar measure over the set of unitaries $U(\Omega)=U_1(\Omega) \otimes \cdots U_N (\Omega)$ induced by the spatial rotation $R(\Omega)$ and integrate over all rotations in $SO(3)$. Each subsystem $\H_k$ in $\H = \H_1 \otimes \cdots \H_N$ transforms under the rotation via an irrep $U_k(\Omega)$ of $SU(2)$ -  $N$ spin-1/2 particles in total. We refer to the application $\rho \rightarrow \G[\rho]$ of this group averaging as `G-twirling'.

\section{Protected Subsystems}\label{Aprotected}

The tensor product representation $U(\Omega)$ is reducible, and the full Hilbert space $\H = \H_1 \otimes \cdots \otimes \H_N$ can be split into irreps of the rotational symmetry group. Specifically, the Schur-Weyl duality \cite{weyl} tells us that the full Hilbert space $\H$ splits into a sum of multiplicity-free irreps of the group $SU(2)\times \S_N$, where $\S_N$ is the discrete permutation group of $N$ elements. In other words, we have that $\H =\oplus_q ( \M_q \otimes \N_q)$ where $\M_q\otimes \N_q \equiv \H_q$ is the \emph{subspace sector} in which $SU(2)$ acts trivially on $\N_q$ and irreducibly on $\M_q$, while $\S_N$ acts trivially on $\M_q$ and irreducibly on $\N_q$.

We shall assume that Eve shares an ordering reference frame with both Alice and Bob, or at the least, she maintains the same ordering of systems when sending multiple copies of the same state. 

Since the subsystems undergo a unitary channel, $\N_q$ are called \textit{protected or decoherence free} virtual subsystems and can hold information that is not erased by rotations.

It is found that 2 spins can protect 1 classical bit, while we need a minimum of 3 spins to encode a single virtual spin subsystem. This follows from the angular momentum addition $ \H=\H_{1/2}\otimes\H_{1/2}\otimes\H_{1/2} =  \M_{3/2}\otimes \N_{3/2} \oplus \M_{1/2} \otimes \N_{1/2}$ with dim$(\M_{3/2})=$4, dim$(\N_{3/2}) =1$, and dim$(\M_{1/2})=$ dim$(\N_{1/2}) =2$.

This decomposition into irreps of $SU(2) \times S_3$ can be described in an orthonormal basis for $\H$ of the form $\{|\lambda, s,p\> \}$ where $\lambda = 3/2, 1/2$ is a symmetry class label and corresponds to the different total angular momentum sectors, while the remaining labels $s$ and $p$ correspond to the action of the rotation group and permutation group respectively. In terms of computational spin bases, we have for the $J=3/2$ sector
\begin{eqnarray}
|\frac{3}{2}, 0,0\> &=& |000\> \non 
|\frac{3}{2}, 1,0\> &=& \frac{1}{\sqrt{3}}(|001\>+|010\>+|100\>) \non 
|\frac{3}{2}, 2,0\> &=& \frac{1}{\sqrt{3}}(|110\>+|101\>+|011\>) \non 
|\frac{3}{2}, 3,0\> &=& |111\> 
\end{eqnarray}
while for the $J=1/2$ sector
\begin{eqnarray}
|\frac{1}{2}, 0,0\> &=&  \frac{1}{\sqrt{2}}(|010\>-|100\>)\non 
|\frac{1}{2}, 0,1\> &=& \frac{1}{\sqrt{6}}(2|001\> - |010\> - |100\>) \non 
|\frac{1}{2}, 1,0\> &=&   \frac{1}{\sqrt{2}}(|011\>-|101\>)\non 
|\frac{1}{2}, 1,1\> &=& \frac{1}{\sqrt{6}}(-2|110\> + |101\> + |011\>) 
\end{eqnarray}

\section{Accessing the protected qubit}\label{virtobs}
We can ask which spin observables must Alice manipulate in order to access the virtual qubit subsystem, protected from the G-twirling. 

Alice has in her possession three spin-1/2 particles, with angular momentum operators $\{ \mathbf{J_1},\mathbf{J_2},\mathbf{J_3} \}$, where for a fixed Cartesian frame with Pauli matrices $\{\sigma^x_a,\sigma^y_a,\sigma^z_a \}$ on spin $a$ we will use the compact vector notation $\mathbf{J}_a:=(J^x_a, J^y_a,J^z_a)=(\frac{1}{2} \sigma^x_a, \frac{1}{2} \sigma^y_a,\frac{1}{2}\sigma^z_a)$.

She knows that for a given state $\rho$ the relevant G-twirling map takes the form
\begin{eqnarray}
\G [ \rho] &=& \D_1 [ \Pi_1 \rho \Pi_1 ] + \D_2 [ \Pi_2 \rho \Pi_2 ]
\end{eqnarray}
where the j=3/2 and j=1/2 sectors are written $\H_1$ and $\H_2=\M_2\otimes \N_2$, with projectors $\Pi_1$ and $\Pi_2$, and where $\D_1$ fully decoheres $\H_1$ while $\D_2$ fully decoheres the virtual subsystem $\M_2$ and leaves $\N_2$ unaltered. 
The observable $S$ that has the sectors $\H_{1,2}$ as eigenspaces is given, in terms of local spin observables, by
\begin{eqnarray}
S&=& \frac{1}{3} (\mathbf{J}_1\cdot\mathbf{J}_2 + \mathbf{J}_2\cdot\mathbf{J}_3 +\mathbf{J}_1\cdot\mathbf{J}_3 )
\end{eqnarray}
where the inner products in $S$ are defined by
\begin{eqnarray}
\mathbf{J}_a\cdot \mathbf{J}_b := (J_a^x \otimes J_b^x+J_a^y \otimes J_b^y+J_a^z \otimes J_b^z)\otimes \I_c.
\end{eqnarray}
We see that $S$ is a relational observable, roughly being the average degree of alignment between the three spins, and is both rotationally invariant and permutationally invariant, as expected from the Schur-Weyl duality.

The physical observables that Alice must measure to access the state on the protected virtual qubit system $\N_2$ are given in terms of the physical spin observables by
\begin{eqnarray}
N_x&=&\frac{1}{\sqrt{3}}( \mathbf{J}_2 -\mathbf{J}_1)\cdot\mathbf{J}_3 \non 
N_y&=& \frac{2}{\sqrt{3}} (\mathbf{J}_1 \times \mathbf{J}_2) \cdot \mathbf{J}_3\non 
N_z&=&\frac{1}{3}( \mathbf{J}_2\cdot\mathbf{J}_3 +\mathbf{J}_1\cdot\mathbf{J}_3 -2\mathbf{J}_1\cdot\mathbf{J}_2). 
\end{eqnarray}
The observables $N_i$ are rotationally invariant, but not invariant under permutations, as expected. 

Using the ``order 2 with order 2'' operator identity
\begin{eqnarray}\label{twotwo}
[\J{a} \cdot \J{c} ,\J{b} \cdot \J{c} ] = i (\J{a} \times \J{b} ) \cdot \J{c}
\end{eqnarray}
together with the ``order 2 with order 3'' identity
\begin{eqnarray}\label{fifthorder}
[\J{b} \cdot \J{c}, (\J{a} \times \J{b}) \cdot \J{c}] &=& \frac{i}{2} \J{a} \cdot (\J{c}- \J{b}) 
\end{eqnarray}
one can readily verify that $\{ N_i \}$ obey the $su(2)$ Lie algebra relations $[N_i, N_j] = i\epsilon_{ijk} N_k$ and also satisfy $N_i^2 =\I_{\H_2}=$ projector onto the $j=1/2$ sector, justifying our labels of $x,y,z$ for the virtual spin observables. 

Furthermore, changes of the order reference of the particles corresponds to the action of the permutation group $\S_3$, which preserves the commutation relations and corresponds to rotations of the virtual Bloch sphere through angles of 120 degrees.

\section{Twirling Alice, Twirling Bob}\label{Atwirling}

In section \ref{entanglement} we are interested in how well local reference frames do in the preservation of bipartite entanglement. The initial product state is assumed to take the form $|\Psi \> = |QRF_A\> \otimes |\varphi \> \otimes |QRF_B \>$, where $|\varphi \>$ is an entangled two spin state and $|QRF_A\>$ and $|QRF_B\>$ consist of $N_A$ and $N_B$ spins respectively. In the absence of a shared classical reference frame for both $A$ and $B$ this means that both sides are G-twirled independently, $|\Psi \> \< \Psi | \rightarrow \rho_{AB}=\G_A \otimes \G_B [|\Psi \> \< \Psi |]$, or more explicitly
\begin{eqnarray}
&&\hspace{-0.4cm} \G_A \otimes \G_B [|\Psi \> \< \Psi |] = \int \hspace{-0.25cm}\int\hspace{-0.1cm} d \Omega d \Omega' U_{AB}(\Omega,\Omega') |\Psi \> \< \Psi |U^\dagger_{AB}(\Omega,\Omega')\nonumber
\end{eqnarray}
where $U_{AB}(\Omega,\Omega') =U(\Omega)^{\otimes N_A+1} \otimes U(\Omega')^{\otimes N_B+1}$ is the unitary corresponding to the rigid rotation of all spins at $A$ through an angle $\Omega$, and the rigid rotation of all spins at $B$ through an angle $\Omega'$.

A brute-force numerical simulation of this G-twirling rapidly gets difficult, and so one must exploit the structure of the decoherence-full/free subsystems to determine how much entanglement is lost for a given pair of local reference frames $QRF_A$ and $QRF_B$.

\end{appendix}

\end{document}